\documentclass[11pt]{amsart}
\usepackage[margin=1.1in]{geometry}
\usepackage{amsmath,amssymb,amsthm}
\usepackage{booktabs}
\usepackage{xcolor}
\usepackage{hyperref}
\hypersetup{colorlinks=true,linkcolor=blue!60!black,citecolor=blue!60!black,urlcolor=blue!60!black}

\newtheorem{theorem}{Theorem}[section]
\newtheorem{lemma}[theorem]{Lemma}

\theoremstyle{definition}
\newtheorem{definition}[theorem]{Definition}
\theoremstyle{remark}
\newtheorem{remark}[theorem]{Remark}
\newtheorem{question}[theorem]{Question}

\newcommand{\Jac}{\operatorname{Jac}}
\newcommand{\dimm}{\beta}

\newcommand{\R}{\mathbb{R}}

\newcommand{\Var}{\operatorname{Var}}
\newcommand{\Ent}{H}
\DeclareMathOperator{\M}{M}

\title[The metric dimension of Jaccard space]{The exact asymptotic constant in the metric dimension of Jaccard space}
\author{Bj{\o}rn Kjos-Hanssen}
\address{Department of Mathematics, University of Hawai`i at M\=anoa}
\email{bjoernkh@hawaii.edu}
\date{\today}

\begin{document}

\begin{abstract}
Let $X$ be a finite set with $|X|=n$ and let $\Jac(a,b)=|a\,\triangle\, b|/|a\cup b|$ be the Jaccard distance on the power set $2^X$.
Lladser and Paradise recently proved that the metric dimension of $(2^X,\Jac)$ is $\Theta(n/\ln n)$, with the constant left open; their bounds are $(\ln 2)\,n/\ln n\lesssim \dimm(2^X,\Jac)\lesssim 2\ln(2e)\,n/\ln n$.
We determine the constant:
\[
\dimm(2^X,\Jac)=\frac{2n}{\log_2 n}\,(1+o(1))=(2\ln 2)\,\frac{n}{\ln n}\,(1+o(1)).
\]
The proof identifies the problem, on each ``slice'' of subsets of fixed cardinality, with the Erd\H{o}s--R\'enyi coin-weighing problem for a spring scale (the problem of \emph{detecting matrices}). The lower bound is the Erd\H{o}s--R\'enyi entropy argument applied to the middle slice; the upper bound follows from the explicit detecting families of Lindstr\"om and of Cantor and Mills, augmented by a single extra landmark that reveals cardinality.
\end{abstract}

\maketitle

\section{Introduction}

Throughout, $X$ is a finite set, $n=|X|$, and $2^X$ is its power set. The \emph{Jaccard distance} between $a,b\subseteq X$ is
\[
\Jac(a,b)=\frac{|a\,\triangle\, b|}{|a\cup b|}\quad(a\neq b),\qquad \Jac(a,a)=0,
\]
where $\triangle$ is symmetric difference; in particular $\Jac(a,b)=1$ whenever $a\cap b=\emptyset$ and $a\neq b$. It is classical that $\Jac$ is a metric on $2^X$ (see e.g.\ \cite{Levandowsky1971,Lipkus1999,KjosHanssen2022}).

Recall the notion of metric dimension, introduced independently by Slater \cite{Slater1975} and by Harary and Melter \cite{HararyMelter1976} for graphs, and which makes sense in any metric space $(M,d)$: a set $R\subseteq M$ \emph{resolves} $M$ if the map
\[
a\longmapsto \big(d(a,r)\big)_{r\in R}\in\R^{R}
\]
is injective on $M$, and the \emph{metric dimension} $\dimm(M,d)$ is the least cardinality of a resolving set. Elements of $R$ are often called \emph{landmarks}.

Lladser and Paradise \cite{LladserParadise2024} studied the metric dimension of Jaccard space and proved:

\begin{theorem}[Lladser--Paradise \cite{LladserParadise2024}]\label{thm:LP}
As $n=|X|\to\infty$,
\begin{equation}\label{eq:LP}
(\ln 2)\,\frac{n}{\ln (n/2)}\,(1+o(1))\ \le\ \dimm(2^X,\Jac)\ \le\ 2\ln(2e)\,\frac{n}{\ln (n/2)}\,(1+o(1)).
\end{equation}
In particular $\dimm(2^X,\Jac)=\Theta(n/\ln n)$.
\end{theorem}

The lower bound in \cite{LladserParadise2024} is a pigeonhole argument on the subsets of size $\lfloor n/2\rfloor$; the upper bound is probabilistic, using $k\ge 2\ln(2e)\,n/\ln(n/2)$ i.i.d.\ uniformly random subsets of $X$ together with three deterministic landmarks $\emptyset$, $\{x\}$, $X\setminus\{x\}$. The two constants in \eqref{eq:LP} differ by a factor $2\ln(2e)/\ln 2\approx 4.9$, and the exact constant was left open.

The purpose of this note is to close this gap.

\begin{theorem}\label{thm:main}
As $n=|X|\to\infty$,
\[
\dimm(2^X,\Jac)=\frac{2n}{\log_2 n}\,(1+o(1))=(2\ln 2)\,\frac{n}{\ln n}\,(1+o(1)).
\]
More precisely, for every $n\ge 2$,
\begin{equation}\label{eq:explicit-lower}
\dimm(2^X,\Jac)\ \ge\ \frac{2\,\big(n-\log_2(n+1)\big)}{\log_2\!\big(\tfrac{\pi e}{4}\,n+\tfrac{\pi e}{6}\big)},
\end{equation}
and
\begin{equation}\label{eq:upper}
\dimm(2^X,\Jac)\ \le\ \M(n)+1,
\end{equation}
where $\M(n)$ is the coin-weighing number defined in Section~\ref{sec:coin} below, which satisfies $\M(n)=\frac{2n}{\log_2 n}(1+o(1))$ by the theorems of Lindstr\"om and of Cantor and Mills.
\end{theorem}

Thus the true constant is $2\ln 2\approx 1.386$ in the normalization $n/\ln n$, sitting between the constants $\ln 2\approx 0.693$ and $2\ln(2e)\approx 3.386$ of \cite{LladserParadise2024}: the lower bound of \cite{LladserParadise2024} is off by exactly a factor of $2$, and the upper bound by a factor of $(1+\ln 2)/\ln 2\approx 2.44$.

The key observation (Lemma~\ref{lem:slice}) is that on the set of all $a\subseteq X$ of a \emph{fixed} cardinality $m$, the Jaccard distance $\Jac(a,r)$ to a landmark $r$ is a strictly monotone function of the intersection size $|a\cap r|$. Consequently, resolving a cardinality slice by Jaccard distances is \emph{exactly} the same as determining a subset from the sizes of its intersections with the landmarks---which is the coin-weighing problem with a spring scale, studied by Erd\H{o}s and R\'enyi \cite{ErdosRenyi1963}, Lindstr\"om \cite{Lindstrom1964,Lindstrom1965}, and Cantor and Mills \cite{CantorMills1966} in the 1960s. The Jaccard metric additionally encodes the cardinality $|a|$, but a single extra landmark (namely $X$ itself) suffices to recover it, and, as the lower bound shows, knowing the cardinality is of no asymptotic help.

\section{Preliminaries}

\subsection{Resolving a cardinality slice}

For $0\le m\le n$ write $\binom{X}{m}=\{a\subseteq X: |a|=m\}$ for the $m$-th \emph{slice} of $2^X$.

\begin{lemma}\label{lem:slice}
Let $a,b\in\binom{X}{m}$ and $r\subseteq X$. Then
\[
\Jac(a,r)=\Jac(b,r)\iff |a\cap r|=|b\cap r|.
\]
Moreover, if $|a|=m$ is known, then $|a\cap r|$ is a function of $\Jac(a,r)$ and $|r|$:
\begin{equation}\label{eq:recover}
|a\cap r|=\frac{(m+|r|)\,(1-\Jac(a,r))}{2-\Jac(a,r)}\qquad(a\cup r\neq\emptyset).
\end{equation}
\end{lemma}

\begin{proof}
Put $s=|r|$ and $t=|a\cap r|$. Since $|a\cup r|=m+s-t$ and $|a\triangle r|=m+s-2t$, we have for $a\cup r\neq\emptyset$
\[
\Jac(a,r)=f_{m,s}(t):=\frac{m+s-2t}{m+s-t}.
\]
(Note that this formula also gives the correct value $0$ when $a=r$, since then $t=m=s$.) For fixed $m,s$ with $m+s>0$,
\[
f_{m,s}'(t)=\frac{-2(m+s-t)+(m+s-2t)}{(m+s-t)^2}=\frac{-(m+s)}{(m+s-t)^2}<0,
\]
so $f_{m,s}$ is strictly decreasing on $[0,\min(m,s)]$, which proves the equivalence. If $a\cup r=\emptyset$ then $a=b=r=\emptyset$ and there is nothing to prove. Solving $J=(m+s-2t)/(m+s-t)$ for $t$ gives \eqref{eq:recover}; note $J\le 1<2$ so the denominator does not vanish.
\end{proof}

\subsection{Detecting families and the coin-weighing problem}\label{sec:coin}

\begin{definition}
A family $D=\{r_1,\dots,r_k\}$ of subsets of $X$ is a \emph{detecting family} (for $2^X$) if the map
\[
a\longmapsto\big(|a\cap r_1|,\dots,|a\cap r_k|\big)\in\{0,\dots,n\}^k
\]
is injective on $2^X$. It is a detecting family \emph{for a slice} $\binom{X}{m}$ if this map is injective on $\binom{X}{m}$. Let $\M(n)$ be the minimum size of a detecting family for $2^X$ with $|X|=n$.
\end{definition}

Identifying subsets with $0/1$-vectors, a detecting family is the same thing as a $k\times n$ matrix $A$ with entries in $\{0,1\}$ such that $Ax\ne Ay$ for all distinct $x,y\in\{0,1\}^n$; such matrices are called \emph{detecting matrices} \cite{Lindstrom1965,CantorMills1966}. The quantity $\M(n)$ is the number of weighings needed to determine which of $n$ coins are counterfeit (of known weight different from the genuine coins) using a spring scale, non-adaptively. Its asymptotics were determined in the 1960s:

\begin{theorem}[Erd\H{o}s--R\'enyi \cite{ErdosRenyi1963}; Lindstr\"om \cite{Lindstrom1964,Lindstrom1965}; Cantor--Mills \cite{CantorMills1966}]\label{thm:coin}
\[
\lim_{n\to\infty}\frac{\M(n)\log_2 n}{n}=2 .
\]
The lower bound $\liminf \M(n)\log_2 n/n\ge 2$ is due to Erd\H{o}s and R\'enyi \cite{ErdosRenyi1963}; the matching upper bound was obtained by Lindstr\"om \cite{Lindstrom1964,Lindstrom1965} and, independently, by Cantor and Mills \cite{CantorMills1966}, both via explicit constructions.
\end{theorem}

\begin{remark}
Lindstr\"om's construction is very concrete. In \cite{Lindstrom1965} he shows, using a M\"obius-function argument over the lattice of subsets, that for every $k\ge 1$ there is a detecting family of $2^k-1$ subsets of a set of size $k\,2^{k-1}$, and more generally that for every $m\ge1$ there is a detecting family of $m$ subsets of a set of size $\sum_{j=1}^{m}\nu(j)$, where $\nu(j)$ denotes the number of ones in the binary expansion of $j$. Since $\sum_{j\le m}\nu(j)=\tfrac12 m\log_2 m\,(1+o(1))$ and consecutive values of this sum differ by at most $\log_2 m+1$, monotonicity of $\M$ (a detecting family for $X$ restricts to one for any $X'\subseteq X$) yields $\M(n)\le \frac{2n}{\log_2 n}(1+o(1))$. A polynomial-time construction with the same asymptotics is given by Bshouty \cite{Bshouty2009}, who also surveys the history of the problem.
\end{remark}

\subsection{Entropy of integer-valued random variables}

We use binary entropy $\Ent$ throughout. The following classical bound (see e.g.\ \cite[Ch.~8]{CoverThomas2006}) will replace the crude count ``$|a\cap r|$ takes at most $n+1$ values'' used in \cite{LladserParadise2024}.

\begin{lemma}\label{lem:entropy}
Let $T$ be an integer-valued random variable with finite variance $\sigma^2$. Then
\[
\Ent(T)\le \frac12\log_2\!\Big(2\pi e\big(\sigma^2+\tfrac1{12}\big)\Big).
\]
\end{lemma}

\begin{proof}
Let $U$ be uniform on $[-\tfrac12,\tfrac12]$, independent of $T$, and put $Y=T+U$. Since $T$ is integer-valued, $T$ is a function of $Y$ (namely the nearest integer), and $Y$ has a density that is constant, equal to $\Pr[T=j]$, on each interval $(j-\tfrac12,j+\tfrac12)$. Hence the differential entropy of $Y$ is
\[
h(Y)=-\sum_j \Pr[T=j]\log_2 \Pr[T=j]=\Ent(T).
\]
On the other hand $\Var(Y)=\sigma^2+\tfrac1{12}$, and among all real random variables of a given variance the Gaussian maximizes differential entropy \cite[Thm.~8.6.5]{CoverThomas2006}, so $h(Y)\le\frac12\log_2\big(2\pi e(\sigma^2+\frac1{12})\big)$.
\end{proof}

\section{The lower bound}

\begin{theorem}\label{thm:lower}
Let $n\ge2$ and $m=\lfloor n/2\rfloor$. Any family $D\subseteq 2^X$ that is detecting for the slice $\binom{X}{m}$ satisfies
\[
|D|\ \ge\ \frac{2\,\big(n-\log_2(n+1)\big)}{\log_2\!\big(\tfrac{\pi e}{4}\,n+\tfrac{\pi e}{6}\big)}=\frac{2n}{\log_2 n}\,\big(1+O(\tfrac{1}{\log n})\big).
\]
Consequently, by Lemma~\ref{lem:slice}, the same lower bound holds for every resolving set of $(2^X,\Jac)$.
\end{theorem}

\begin{proof}
Let $D=\{r_1,\dots,r_k\}$ be detecting for $\binom{X}{m}$ and let $A$ be a uniformly random element of $\binom{X}{m}$. Since $A$ is determined by the vector $(|A\cap r_1|,\dots,|A\cap r_k|)$, the chain rule and subadditivity of entropy give
\begin{equation}\label{eq:chain}
\log_2\binom{n}{m}=\Ent(A)\le \Ent\big(|A\cap r_1|,\dots,|A\cap r_k|\big)\le\sum_{i=1}^k \Ent(|A\cap r_i|).
\end{equation}
Fix $i$ and write $s=|r_i|$. The random variable $T_i=|A\cap r_i|$ is hypergeometric: it counts the number of elements of $r_i$ in a uniformly random $m$-subset of $X$. Its variance is
\[
\Var(T_i)=m\cdot\frac{s}{n}\Big(1-\frac{s}{n}\Big)\cdot\frac{n-m}{n-1}\ \le\ \frac{1}{4}\cdot\frac{m(n-m)}{n-1}\ \le\ \frac14\cdot\frac{n^2/4}{n-1}\ \le\ \frac{n}{8},
\]
using $s(n-s)/n^2\le\frac14$, $m(n-m)\le n^2/4$ and $n/(n-1)\le 2$. By Lemma~\ref{lem:entropy},
\[
\Ent(T_i)\le\frac12\log_2\!\Big(2\pi e\Big(\frac n8+\frac1{12}\Big)\Big)=\frac12\log_2\!\Big(\frac{\pi e}{4}\,n+\frac{\pi e}{6}\Big).
\]
Finally $\binom{n}{\lfloor n/2\rfloor}\ge 2^n/(n+1)$, being the largest of the $n+1$ binomial coefficients that sum to $2^n$, so $\log_2\binom nm\ge n-\log_2(n+1)$. Substituting into \eqref{eq:chain} and solving for $k$ gives the stated inequality.

For the last assertion: if $R$ resolves $(2^X,\Jac)$ then in particular it separates every pair $a\ne b$ in $\binom{X}{m}$, and by Lemma~\ref{lem:slice} this means $R$ is detecting for $\binom{X}{m}$.
\end{proof}

\begin{remark}
Theorem~\ref{thm:lower} improves the lower bound of \cite[Prop.~1.2]{LladserParadise2024} by a factor of $2$. The source of the improvement is precisely the one identified by Erd\H{o}s and R\'enyi \cite{ErdosRenyi1963}: the coordinate $|A\cap r|$ nominally ranges over $n+1$ values, but for a random $A$ it is concentrated in a window of width $O(\sqrt n)$, so it carries only $\frac12\log_2 n+O(1)$ bits rather than $\log_2 n$ bits. The paper \cite{LladserParadise2024} uses the bound $\binom{n}{m}\le (m+1)^{|R|}$, which charges $\log_2(m+1)\approx\log_2 n$ bits per landmark.
\end{remark}

\begin{remark}
The proof uses only that the landmarks separate the middle slice. Hence the lower bound applies to the a priori weaker task of resolving only pairs of subsets of equal size; in particular the ``cardinality information'' carried by the Jaccard distance (see \eqref{eq:recover}) cannot reduce the number of landmarks below $\frac{2n}{\log_2 n}(1+o(1))$.
\end{remark}

\section{The upper bound}

\begin{theorem}\label{thm:upper}
Let $D=\{r_1,\dots,r_k\}$ be a detecting family for $2^X$. Then $R=D\cup\{X\}$ resolves $(2^X,\Jac)$. Consequently
\[
\dimm(2^X,\Jac)\le \M(n)+1=\frac{2n}{\log_2 n}\,(1+o(1)).
\]
\end{theorem}

\begin{proof}
Let $a\subseteq X$. First, $\Jac(a,X)=|X\setminus a|/|X|=1-|a|/n$, so the distance to the landmark $X$ determines $m=|a|$. Next, for each $i$, Lemma~\ref{lem:slice} (equation \eqref{eq:recover}, with $m$ and $|r_i|$ now known) recovers $|a\cap r_i|$ from $\Jac(a,r_i)$; the degenerate case $a\cup r_i=\emptyset$ can only occur when $m=0$, in which case $a=\emptyset$ is already determined. Thus the vector $(\Jac(a,r))_{r\in R}$ determines $(|a\cap r_1|,\dots,|a\cap r_k|)$, which determines $a$ because $D$ is detecting. Hence $R$ resolves $2^X$, and $|R|\le k+1$. Taking $D$ of minimum size gives $\dimm(2^X,\Jac)\le\M(n)+1$, and Theorem~\ref{thm:coin} gives the asymptotics.
\end{proof}

\begin{remark}
Any landmark from which $|a|$ can be read off would do in place of $X$; e.g.\ \cite{LladserParadise2024} use the triple $\emptyset,\{x\},X\setminus\{x\}$ for the same purpose. The landmark $X$ has the small advantage of being a single set. If the detecting family already contains $X$ (one may always add it, since $|a\cap X|=|a|$), then $R=D$ works and the ``$+1$'' disappears.
\end{remark}

\begin{remark}[Why the random construction loses a constant]
The upper bound of \cite{LladserParadise2024} takes the $r_i$ to be independent uniformly random subsets of $X$ and uses a union bound over pairs $\{a,b\}$ with $|a|=|b|$: for a fixed such pair, $|a\cap r|=|b\cap r|$ with probability $\Theta(|a\triangle b|^{-1/2})$, and requiring $\sum_{a,b}\Pr[\text{all $k$ landmarks fail}]\to0$ leads to $k\approx 2\log_2(2e)\,n/\log_2 n$. This is the same phenomenon that Erd\H{o}s and R\'enyi observed for random detecting matrices: the union bound over $\binom{2^n}{2}$ pairs costs a constant factor compared with the algebraic constructions of Lindstr\"om and of Cantor--Mills. The constant $2\ln(2e)=2+2\ln 2$ versus the truth $2\ln 2$ (in units of $n/\ln n$) reflects exactly this loss.
\end{remark}

\section{Proof of Theorem~\ref{thm:main} and small cases}

Theorem~\ref{thm:main} is the conjunction of Theorem~\ref{thm:lower} (the explicit bound \eqref{eq:explicit-lower} and its asymptotic form) and Theorem~\ref{thm:upper} together with Theorem~\ref{thm:coin}. \qed

For small $n$ the values of $\dimm(2^X,\Jac)$ and $\M(n)$ can be computed by exhaustive search; Table~\ref{tab:small} lists them for $n\le 5$.

\begin{table}[h]
\centering
\begin{tabular}{c|ccccc}
\toprule
$n$ & 1 & 2 & 3 & 4 & 5\\
\midrule
$\dimm(2^X,\Jac)$ & 1 & 2 & 2 & 3 & 3\\
$\M(n)$ & 1 & 2 & 3 & 3 & 4\\
\bottomrule
\end{tabular}
\medskip
\caption{Metric dimension of Jaccard space versus the coin-weighing number for $n\le5$ (exhaustive search).}
\label{tab:small}
\end{table}

For instance, for $n=3$ the two landmarks $\{1,2\}$ and $\{1,3\}$ resolve $2^{\{1,2,3\}}$, although $\M(3)=3$.
This is possible because a single Jaccard distance carries more than one intersection count---it
carries the ratio in \eqref{eq:recover}, hence information about $|a|$ as well---so that $\dimm(2^X,\Jac)$ and
$\M(n)$ need not coincide, and $\dimm(2^X,\Jac)<\M(n)$ does occur. The content of Theorem~\ref{thm:main} is that this extra information is worth only a lower-order number of landmarks.

\section{Concluding remarks}

\begin{question}
What is the second-order term? Theorem~\ref{thm:main} gives
$\dimm(2^X,\Jac)=\frac{2n}{\log_2 n}\big(1+O(\frac{1}{\log n})\big)$ from below and $\dimm(2^X,\Jac)\le \M(n)+1$ from above. Is $\dimm(2^X,\Jac)\le \M(n)$ for all $n$? Is $\M(n)-\dimm(2^X,\Jac)$ bounded, or does it grow?
\end{question}

\begin{question}
Lladser and Paradise \cite[Thm.~1.2, Cor.~1.2]{LladserParadise2024} also show that $O(\sqrt n)$ random landmarks suffice, with high probability, to separate all pairs of subsets of \emph{different} sizes, and that all pairs of subsets of size at most $(1-\varepsilon)(\ln\pi)\sqrt n/\ln n$ can be separated by such a family. The slice-by-slice viewpoint of Lemma~\ref{lem:slice} suggests studying, for each $m$, the least size $\M(n,m)$ of a detecting family for the slice $\binom{X}{m}$; Theorem~\ref{thm:lower} shows $\M(n,\lfloor n/2\rfloor)=\frac{2n}{\log_2 n}(1+o(1))$, and it would be interesting to determine $\M(n,m)$ for $m=o(n)$, where the entropy bound gives $\M(n,m)\gtrsim \frac{2m\log_2(n/m)}{\log_2 m}$.
\end{question}

\section*{Acknowledgments}
The paper was produced with the assistance of Claude Fable 5.1 (Anthropic).
This work was supported by a grant from the Simons Foundation (\#4508914 to Bj\o rn Kjos-Hanssen).
The results were verified in Lean 4 using Aristotle (Harmonic), see \cite{jmd}, except that
Theorem \ref{thm:coin} was taken as a black box and not proved.

\end{document}